\documentclass[11pt,oneside]{amsart}

\usepackage[nocompress]{cite} 

\usepackage{graphicx}
\usepackage{subfig}
\usepackage{tabularx}

\usepackage{caption}
\usepackage{amsaddr}
\usepackage{color}
\usepackage{enumerate}
\usepackage{bbm}
\usepackage[mathscr]{eucal}
\usepackage{verbatim}
\usepackage{wrapfig}
\usepackage{amsmath}%
\usepackage{MnSymbol}%
\usepackage{wasysym}%
\usepackage{hyperref}

\usepackage{algorithmic}
\usepackage{algorithm}

\usepackage[margin=1in]{geometry}

\newtheorem{definition}{Definition}[section]

\newtheorem{theorem}{Theorem}[section]
\newtheorem{lemma}[theorem]{Lemma}

\theoremstyle{remark}

\numberwithin{equation}{section}

\begin{document}

\title[EMDUnifrac]{EMDUnifrac: Exact linear time computation of the Unifrac metric and identification of differentially abundant organisms}
\author[McClelland \& Koslicki]{Jason McClelland\,$^{1}{}^*$, David Koslicki\,$^{1}$}
\thanks{${}^*$\url{mcclellj@science.oregonstate.edu}}
\address{$^{1}$Mathematics Department, Oregon State University, Corvallis OR}

\maketitle

\begin{abstract}

Both the weighted and unweighted Unifrac distances have been very successfully employed to assess if two communities differ, but do not give any information about \textit{how} two communities differ. We take advantage of recent observations that the Unifrac metric is equivalent to the so-called \textit{earth mover's distance} (also known as the Kantorovich-Rubinstein metric) to develop an algorithm that not only computes the Unifrac distance in linear time and space, but also simultaneously finds which operational taxonomic units are responsible for the observed differences between samples. This allows the algorithm, called EMDUnifrac, to determine \textit{why} given samples are different, not just \textit{if} they are different, and with no added computational burden. EMDUnifrac can be utilized on any distribution on a tree, and so is particularly suitable to analyzing both operational taxonomic units derived from amplicon sequencing, as well as community profiles resulting from classifying whole genome shotgun metagenomes.
The EMDUnifrac source code (written in python) is freely available at: 
\url{https://github.com/dkoslicki/EMDUnifrac}.


\end{abstract}

\section{Introduction}
An important first step in comparative microbial ecology studies is the assessment of if and how two communities of microorganisms differ. Unifrac \cite{lozupone2007quantitative,lozupone2005unifrac,hamady2010fast}, in its various implementations, is a commonly utilized distance metric that quantifies if two communities do indeed differ. In the field of metagenomics, this phylogentic-aware distance has been used to effectively cluster many 16S rRNA samples and distinguish between them based on a given environmental factor \cite{ley2006ecological,frank2007molecular,rawls2006reciprocal}. However, a recognized disadvantage to the Unifrac distance is that it only quantifies \textit{if} two communities differ and gives no indication of \textit{how} they differ \cite{white2009statistical}. Typically, to answer the question of how two communities differ, further statistical or computational methods are employed \cite{white2009statistical,schloss2006introducing,wooley2010primer,parks2010identifying}. 

In this article, we demonstrate that in viewing the Unifrac distance as the so-called Kantorovich-Rubinstein metric (also known as the earth mover's distance \cite{rubner2000earth}), one can obtain exactly \textit{how} two communities differ and which operational taxonomic units (OTUs) or taxa are responsible for the observed Unifrac distance. This equivalence between the Unifrac distance and the earth mover's distance was demonstrated recently  \cite{evans2012phylogenetic} and while this equivalence greatly improved the understanding of the Unifrac distance, the authors of \cite{evans2012phylogenetic} were primarily concerned with assessing statistical significance of Unifrac distances and not with detailing how this view can be used to returning differentially abundant OTUs.

We begin first by detailing how using the earth mover's distance to compute the Unifrac distance can identify differentially abundant OTUs. We then introduce a linear time algorithm, called EMDUnifrac, that computes the Unifrac distance and also returns the differentially abundant OTUs that contributed to this distance. Finally, after demonstrating its usefulness on previously published biological data, we prove the correctness of EMDUnifrac and calculate its time and space complexity.



\section{Identifying differentially abundant OTUs}
To demonstrate how viewing the Unifrac distance as the earth mover's distance (EMD) identifies differentially abundant OTUs, we first need to define the EMD. We focus here on the weighted (normalized) Unifrac distance, as the unweighted Unifrac can be obtained by appropriately modifying the underlying distributions utilized.

\begin{figure}[!htpb]
\centerline{\includegraphics[scale=0.6,trim={1.5cm 2.5cm 2cm 2cm},clip]{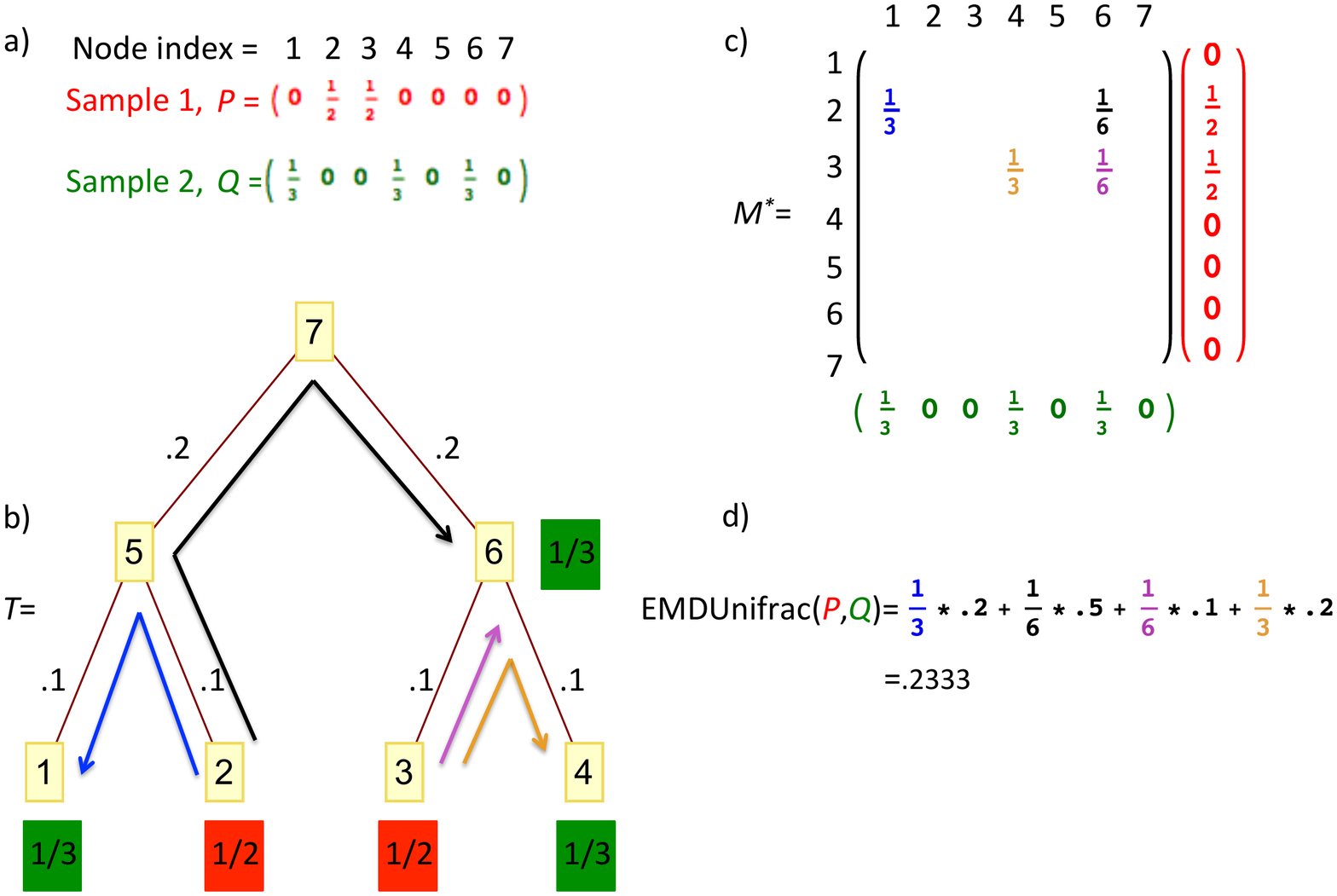}}
\caption{Visualization of computing the Unifrac distance using the earth mover's distance. The abundances of the individual samples are pictured in part a), indexed by the nodes of the tree $T$ which is pictured in part b) along with branch lengths. The minimizing flow $M^*$ is pictured in part c) and the colors of the entries of $M^*$ correspond to the colored arrows on the tree. Part d) contains the computation of ${\rm EMDUnifrac}(P,Q)$ with this minimizing flow. }\label{fig:Example}
\end{figure}

Given two sample communities and the associated abundances of microorganisms therein, we can associate to these a phylogenetic tree $T$ and two probability distributions $P$ and $Q$ that represent the fraction of a given sample that appears at each node of the phylogenetic tree (not necessarily restricted to the leaves). As the phylogenetic tree $T$ has associated branch lengths, we can find the minimal distance between any two nodes of the tree. Let $D$ be the matrix of all pairwise distances between nodes in $T$. We use the notation $\Gamma(P,Q)$ to describe the space of all ways in which one community can be transformed into the other. The elements $M\in \Gamma(P,Q)$ are referred to as \textit{flows} and are matrices indexed by the nodes in the tree $T$ with the stipulation that the row sums of $M$ are equal to $P$ and the column sums of $M$ are equal to $Q$. The $(i,j)^{\rm th}$ entry of such an $M$ indicates that a total abundance of $M_{i,j}$ has been moved from node $i$ in the sample $P$ to node $j$ in sample $Q$. With these conventions, we can define the earth mover's distance on this tree, which we refer to herein as EMDUnifrac:
\begin{equation}
\label{eq:EMD}
{\rm EMDUnifrac}(P,Q) = \underset{M\in \Gamma(P,Q)}{\rm minimize} \sum_{i,j\in T} D_{i,j} M_{i,j}.
\end{equation}
Informally, the quantity ${\rm EMDUnifrac}(P,Q)$ represents the minimum amount of ``work'' required to transform the distribution of one sample $P$ into the distribution of the other sample $Q$ along the phylogenetic tree.

It has been previously shown, using different notation, that ${\rm EMDUnifrac}(P,Q)$ is equivalent to the weighted (normalized) Unifrac distance \cite{evans2012phylogenetic}. Equivalence can be shows for the unweighted Unifrac distance by modifying the distributions $P$ and $Q$ to be binary vectors on the same original support and redefining the space of all flows $\Gamma(P,Q)$. A toy example is given in Figure \ref{fig:Example} that details the previously defined quantities.

We concentrate on the flow $M^*$ that minimizes the right hand side of the expression in \eqref{eq:EMD} and call this the \textit{minimizing flow}. This matrix represents where the abundance of one sample was moved when it was being transformed into the other sample, and this quantity precisely describes \textit{how} the two samples differ and which OTUs contributed to the computed Unifrac value. For example, in Figure \ref{fig:Example}, the entry $M^*_{2,1}=\frac{1}{3}$ indicates that $1/3^{\rm rd}$ of the abundance of the first sample was moved from node 2 to node 1. A little care must be taken, though, as it is not guaranteed that there is one \textit{unique} minimizing flow. In all cases, the elements on the diagonal of any minimizing flow can be ignored (as this only indicates the abundances that were the same between the two samples). However, we can define a vector indexed by the edges of our phylogenetic tree called the \textit{differential abundance vector}, which is the same no matter which minimizing flow is chosen. Letting $E$ denote the edges of our phylogenetic tree, $T_e$ the nodes of the subtree below an edge $e \in E$ and $T_{e'}$ the remaining nodes of $T$, so that $T = T_e \cup T_{e'}$, we have that ${\rm DiffAbund}(e) = l(e)\sum_{i \in T_e}\sum_{j \in T_{e'}} M_{i,j}-M_{j,i}$. 
Normalizing this vector so its sum is 1 leads to the following biological interpretation:
\begin{quote}
The normalized differential abundance vectors indicate which taxa contributed to the Unifrac distance and by what percentage.
\end{quote}

For typical metagenomics and metatranscriptomic studies, the distributions $P$ and $Q$ are supported on the leaves of the tree $T$. In this case, minimizing flows and differential abundance vectors can be defined for all nodes, as well as at any fixed taxonomic rank simply by summing over the lower taxa. Figure \ref{fig:RealFlow} gives such an example at the phylum level.

\section{Application to real data}
To demonstrate the utility of EMDUnifrac on real data, we evaluate it on the 16S rRNA data from a previous study\cite{willing2010pyrosequencing}. This data consists of 454 pyrosequenced fecal samples from a cohort of 40 twin pairs. The RDPII \cite{maidak2001rdp} and BLAST \cite{altschul1990basic} classifications were accessed via QIIME/QIITA \cite{caporaso2010qiime}. For simplicity, we focus here on the phylum level, and so summed these classifications to this level. We selected a subset of the data consisting of 49 healthy samples and 16 ulcerative colitis samples and used the silva taxonomic tree\cite{yilmaz2013silva} for the EMDUnifrac computation.

We evaluated the EMDUnifrac algorithm on all 2,080 pairs of samples and performed a principle coordinate analysis (PCoA) on the resulting distance matrix (disregarding the flows for each pair). The result of this is contained in part (A) of Figure \ref{fig:RealFlow}. Next, we combined all the healthy samples and combined all the ulcerative colitis samples and evaluated EMDUnifrac on these two combined samples. The returned minimizing flow is depicted in part (B) of Figure \ref{fig:RealFlow}. The corresponding differential abundance vector is shown in part (C). Even though upon visual inspection, the PCoA plot in part (A) does not show much distinction between healthy and ulcerative colitis samples (compare to the similar plot contained in Figure 2 of \cite{willing2010pyrosequencing}), the differential abundance vector immediately leads to the conclusion that the ulcerative colitis samples are primarily enriched for Actinobacteria and Proteobacteria, while being deficient in Bacteroidetes. This is consistent with other studies where the same trend was observed in irritable bowel disease subjects, but using more intricate analysis techniques \cite{frank2007molecular,spor2011unravelling,manichanh2012gut}, and demonstrates how utilizing the minimizing flow results in more information than simply using a dimension reduction technique (here PCoA) on the pairwise Unifrac distances.

\begin{figure}
\def\tabularxcolumn#1{m{#1}}
\begin{tabularx}{\linewidth}{@{}cXX@{}}
\begin{tabular}{cc}
\subfloat[]{\includegraphics[scale=0.4,trim={.9cm 6.5cm 2cm 6.5cm},clip]{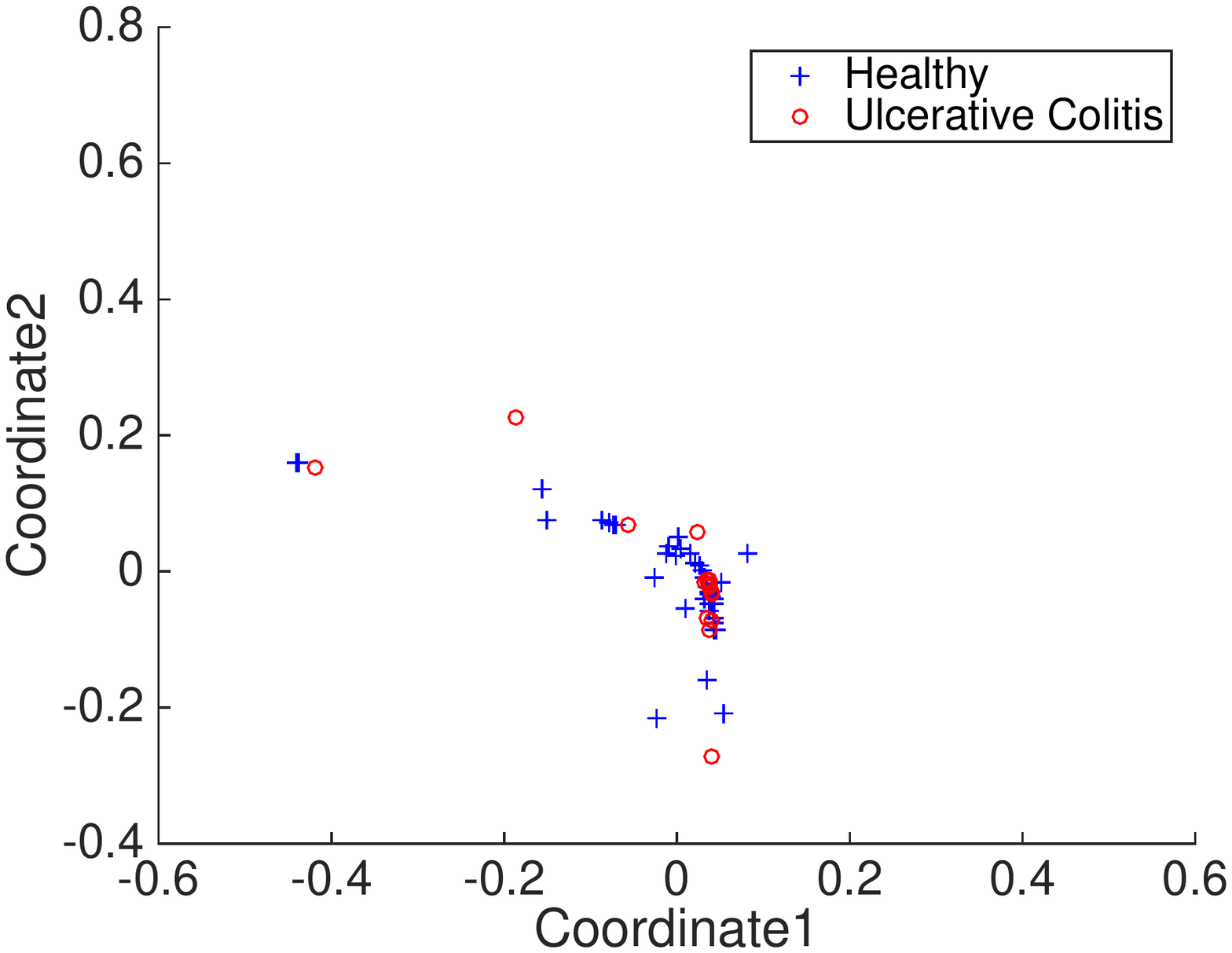}} 
   & \subfloat[]{\includegraphics[scale=0.4,trim={1cm 6.5cm 2cm 6.5cm},clip]{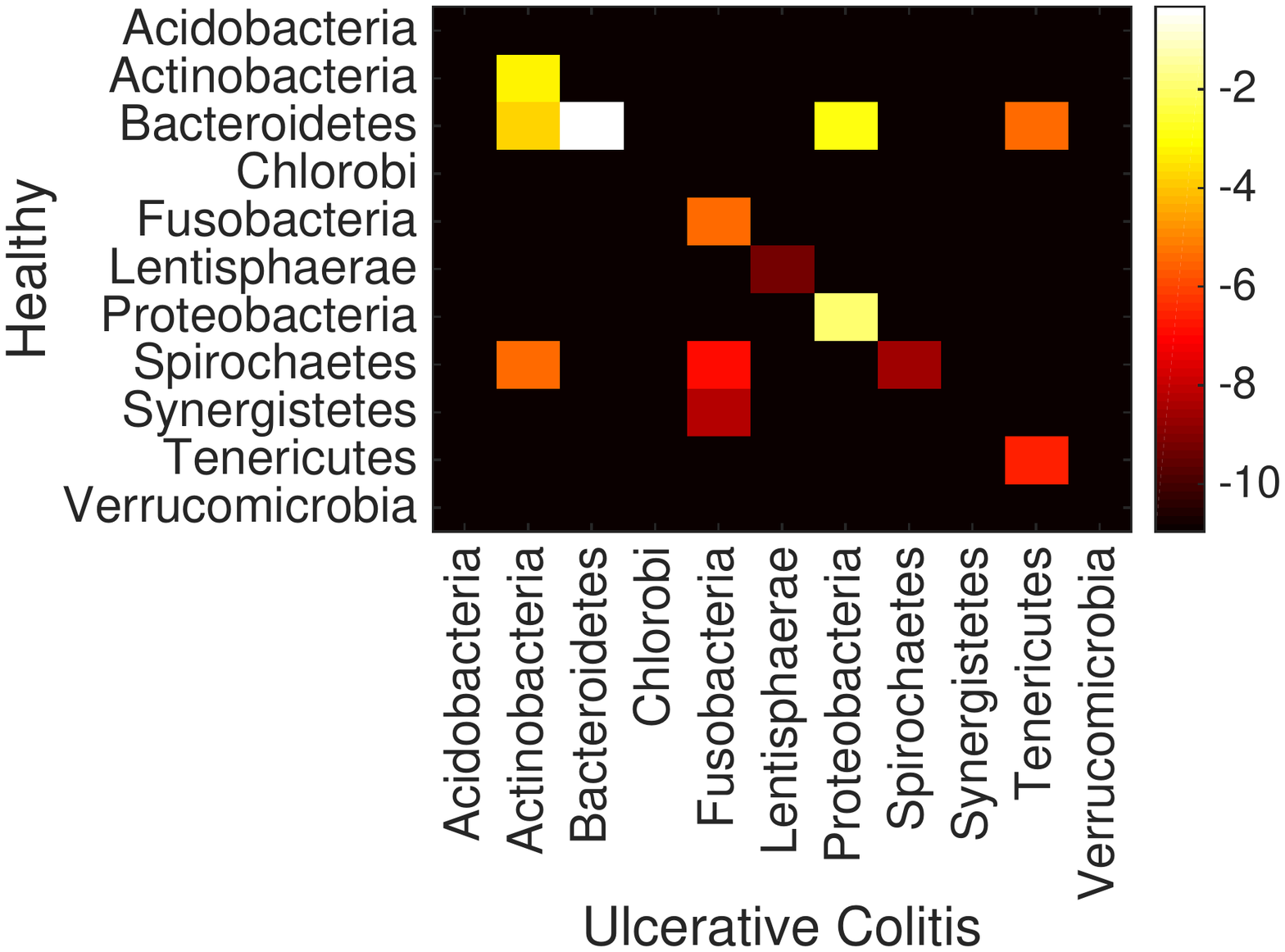}}\\
\multicolumn{2}{c}{\subfloat[]{\includegraphics[width=4in,trim={0cm 0cm 0cm 0cm},clip]{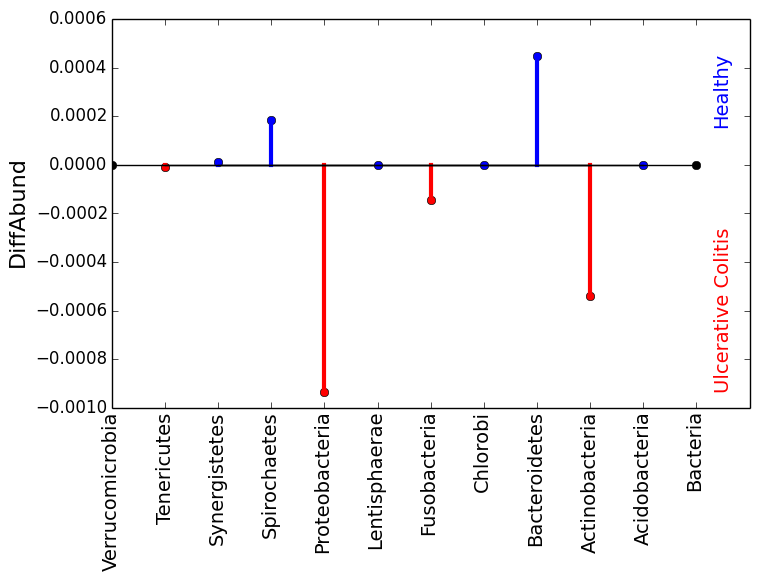}}}
\end{tabular}
\end{tabularx}
\caption{Results on the real data. Part (A) is the PCoA plot of the EMDUnifrac distance matrix between all pairs of samples analyzed. Compare to the similar plot in Figure 2 of \cite{willing2010pyrosequencing}. Part (B) contains a heat map of the unique minimizing flow for the combined healthy and ulcerative colitis samples. This heat map is scaled logarithmically for visualization purposes. Part (C) depicts the differential abundance vector between the combined healthy and Ulcerative Colitis samples and indicate which organisms are differentially abundant in the samples, demonstrating usefulness over the PCoA plot in part (A).}
\label{fig:RealFlow}
\end{figure}

\subsection{Speed comparison to Unifrac}
As modern comparative metagenomics studies often perform all pair-wise Unifrac distance computations for datasets consisting of tens to thousands of samples, it is important to compute such distances in an efficient manner. We show in Theorem \ref{thm:speed} below that our Algorithm \ref{alg:EMDUnifrac} to compute EMDUnifrac runs in space and time complexity linear in the total support of the input vectors (so less than or equal to the number of nodes in the tree). To assess practical performance of Algorithm \ref{alg:EMDUnifrac}, we compared it to the fastest previous implementation of Unifrac, called FastUnifrac \cite{hamady2010fast}. We randomly generated trees (using the ete2 toolkit \cite{ete3}) with the number of leaf nodes ranging from 10 to 90,000. We then randomly produced pairs of distributions on the leaves using an exponential distribution with scale parameter 1. Importantly, EMDUnifrac can handle distributions with weights on leaf nodes as well as internal nodes while FastUnifrac only allows distributions with weights on the leaf nodes. We performed 10 replicates for each number of tree leaves and 10 replicates for each tree topology. Using the same fixed computational resources, we then ran FastUnifrac, EMDUnifrac in a mode that computes and returns the computed flow, and EMDUnifrac in a mode that just calculates the distance (and does not return an optimal flow, returning identical output to FastUnifrac). The average timings (over each number of tree leaves) are depicted in Figure \ref{fig:timing}. These results indicate that in either mode, EMDUnifrac is more computationally efficient than FastUnifrac, and when just the resulting distance is desired, \mbox{EMDUnifrac} takes less than half a second to run, even on trees with 90,000 leaves (noting that our implementation is a non-optimized, Python implementation). 
\begin{figure}[!hb]
\centerline{\includegraphics[width=5in,trim={0cm 0cm 0cm 0cm},clip]{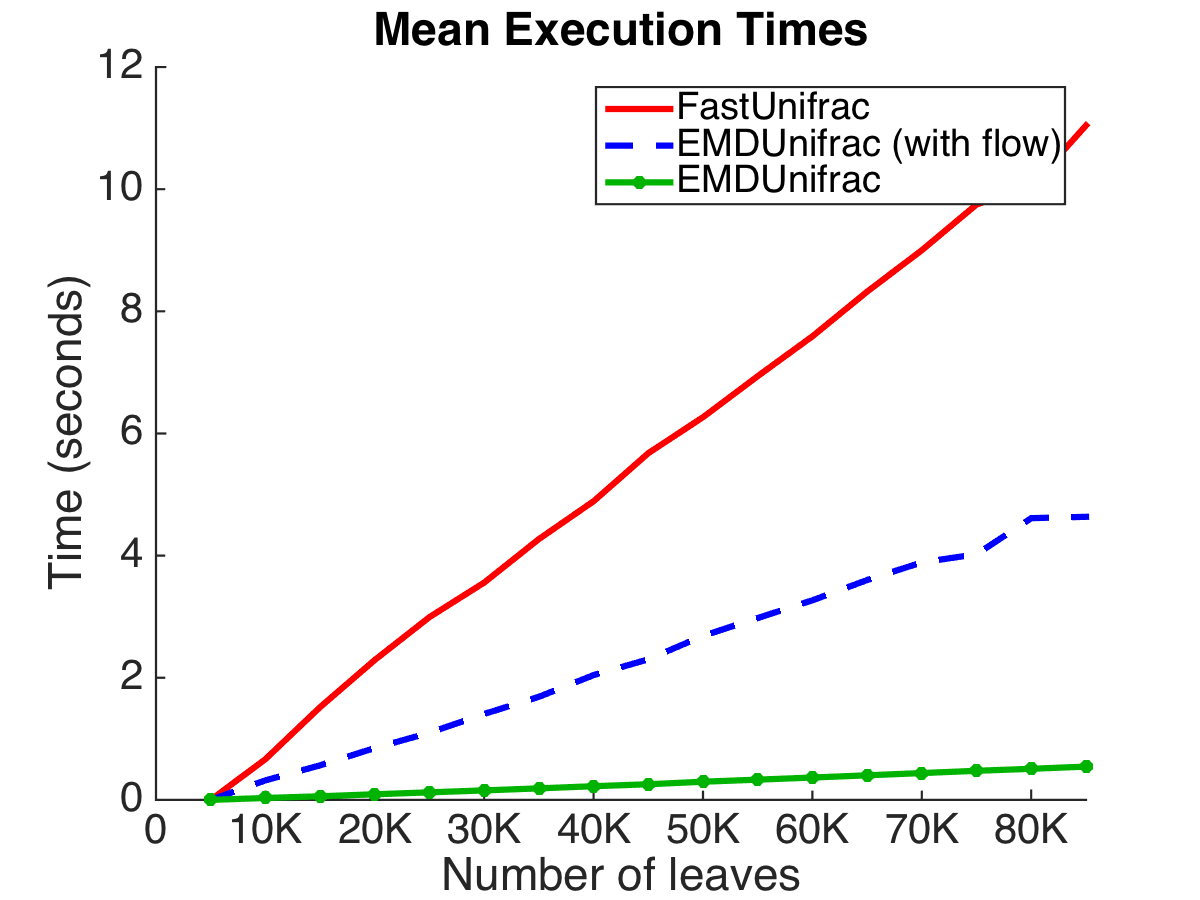}}
\caption{Speed comparison of FastUnifrac to EMDUnifrac (while also returning the minimizing flow) and EMDUnifrac (while returning just the distance). Trees are generated with random topology and abundances are random realizations of an exponential distribution and are supported on the leaves.}
\label{fig:timing}
\end{figure}



\section{Proof of correctness}
In this section, we detail our algorithm to compute EMDUnifrac, prove its correctness, and assess its computational complexity.
\subsection{Definitions and algorithm}

We begin with some definitions. Let $P$ and $Q$ be probability distributions on a tree $T$ with distance matrix $D_{i,j}$ and edge set $E$. Recall that $\Gamma(P,Q)$ is the set of all flows from $P$ to $Q$ in $T$. By an abuse of notation, we write $i \in T$ to denote a vertex of our tree. For such a vertex $i \in T$ we will say $i$ is a \textit{source} if $P_i \geq Q_i$ and say $i$ is a \textit{sink} otherwise. Let $T_{source}$ and $T_{sink}$ denote the sets of sources and sinks, respectively.

Next, we select an arbitrary vertex of $T$ and distinguish it as the root $\rho$ of $T$. The choice is a convenience of notation. For each $i \in T$ let $a(i)$ be the unique neighbor of $i$ in $T$ which lies on the path from $i$ to $\rho$ in $T$. Thus the edges of $T$ are determined by the set of ordered pairs $(i,a(i)))$ for $i \in T$. Let $e_i$ denote the edge $(i,a(i))$. As $T$ is a tree, each edge $e \in E$ is a bridge. Thus its removal partitions the vertices into two disjoint subsets. We denote the subset containing $\rho$ by $T_e$ and the other by $T'_e$. Let $l: E \rightarrow \mathbb{R}_{\geq0}$ define a set of edge weights or lengths on $E$. For $i,j \in T$, define $\pi(i,j)$ to be the set of edges comprising the unique minimal path from $i$ to $j$ in $T$ and let $D_{i,j}= \sum_{e \in \pi(i,j)}{l(e)}$ be the distance from $i$ to $j$ in $T$.

The pseudocode for EMDUnifrac is contained in Algorithm \ref{alg:EMDUnifrac}. Intuitively, the algorithm begins at the leaves of the tree and ``pushes" mass toward the root; satisfying the sources and sinks for each subtree encountered in the progression. The matrix $G$ tracks the mass still needed to be moved to or from each vertex by the algorithm, while the vector $w$ tracks the length of paths traversed by mass at each step.

To implement EMDUnifrac, we first choose an ordering on the set of vertices of $T$ such that for $i,j \in T$, $i$ is an element of the path from $j$ to $\rho$ only if $i \geq j$. A natural such ordering is defined by partitioning the vertices of $T$ by the number of edges in the path to $\rho$, and then ordering vertices such that increasing indices correspond to decreasing path lengths to $\rho$.

We then let $G$ and $M$ be a pair of matrices whose rows and columns are indexed by the vertices of $T$ with respect to an ordering as above. Let $G_{i,\cdot}$ denote the $i$-th row of the matrix $G$. Initialize both $G$ and $M$ to be the zero matrix. Let $w$ be a vector indexed by the vertices of $T$, initialized to be the zero vector. For any vector $u$, define ${\rm skel}(u)$ to be the binary vector of the same dimension as $u$ such for all $i$, ${\rm skel}(u(i)) = 1$ if $u(i) \neq 0$ and ${\rm skel}(u(i)) = 0$ otherwise.

\label{section:Algorithm}
\begin{algorithm}[ht!]
\caption{: ${\rm EMDUnifrac}$\label{alg:EMDUnifrac}}
\begin{flushleft}
\hspace{0ex}\mbox{\textit{Input:}}
\begin{algorithmic}
\STATE $P,Q,\rho, T, E = \{i,a(i)\} \mbox{ for } i\in T, l$
\end{algorithmic}
\hspace{0ex}\mbox{\textit{Initialization:}}
\begin{algorithmic}
\STATE $M,G = \mathbf{0}$
\STATE ${\rm EMDUnifrac}(P,Q) = 0$
\STATE ${\rm DiffAbund} = \vec{0}$
\end{algorithmic}
\hspace{0ex}\mbox{\textit{Iterations:}}
\begin{algorithmic}[1]
\FOR{$i = 1,..., |T|$}
\STATE $M_{i,i} = \min(P_i,Q_i)$
\STATE $G_{i,i} = P_i - Q_i$
\FOR{$j$ such that $G_{i,j} > 0$}
\FOR{$k$ such that $G_{i,k} < 0$}
\STATE $M_{j,k} = \min(G_{i,j},-G_{i,k})$
\STATE $G_{i,j} = G_{i,j} - M_{j,k}$
\STATE $G_{i,k} = G_{i,k} + M_{j,k}$
\STATE ${\rm EMDUnifrac}(P,Q) = {\rm EMDUnifrac}(P,Q) + (w_j+w_k) M_{j,k}$
\ENDFOR
\ENDFOR
\STATE $G_{a(i),\cdot} = G_{a(i),\cdot} + G_{i,\cdot}$
\STATE ${\rm DiffAbund}({(i,a(i))}) = l(i,a(i))\sum_{t\in T}{G_{i,t}}$
\STATE $G_{i,\cdot} = \vec{0}$
\STATE $w = w + l(i,a(i)){\rm skel}(G_{i,\cdot})$
\ENDFOR
\end{algorithmic}
\hspace{0ex}\mbox{\textit{Output:}}
\begin{algorithmic}
\STATE $M$, ${\rm EMDUnifrac}(P,Q)$ and ${\rm DiffAbund}$
\end{algorithmic}
\end{flushleft}
\end{algorithm}

\subsection{Proof of correctness}

We first prove an alternate characterization of the earth movers distance for probability distributions on a tree $T$. 

\begin{lemma}
 We have that $${\rm EMDUnifrac}(P,Q) = \min_{M \in \Gamma(P,Q)}{\sum_{e \in E}\sum_{i\in{T_e}}\sum_{j\in {T'_e}}{l(e)\left(M_{i,j}+M_{j,i}\right)}}.$$
\end{lemma}

\begin{proof}
Let $1_{\pi(i,j)}(e):E \rightarrow \{0,1\}$ be the indicator function of the path from $i$ to $j$ in $T$. That is, $1_{\pi(i,j)}(e) = 1$ if $e$ is an edge in the path from $i$ to $j$ and is $0$ otherwise. We then have that for any flow $M \in \Gamma(P,Q)$
\begin{align}
\sum_{i,j \in T}{D_{i,j}M_{i,j}} &= \sum_{i\in T}\sum_{j\in T}\left(\sum_{e\in E}{l(e)1_{\pi(i,j)}(e)}\right)M_{i,j} \label{lem1_1}\\
	&= \sum_{e\in E}\sum_{i\in T}\sum_{j \in T}{l(e)1_{\pi(i,j)}(e)}M_{i,j}\label{lem1_2}\\
	&= \sum_{e\in E}\:\sum_{\substack{i\in \\ T_e \cup T'_e}}\:\sum_{\substack{i\in \\ T_e \cup T'_e}}{l(e)1_{\pi(i,j)}(e)}M_{i,j}\label{lem1_3}\\
	&= \sum_{e\in E}\left(\sum_{i\in T_e}\sum_{j \in T'_e}l(e)M_{i,j}+\sum_{i\in T'_e}\sum_{j \in T_e}l(e)M_{i,j}\right)\label{lem1_4}\\
	&= \sum_{e\in E}\sum_{i\in T_e}\sum_{j \in T'_e}l(e)\left(M_{i,j}+M_{j,i}\right).\label{lem1_5}
\end{align}

The above equalities are justified as follows. To begin, (\ref{lem1_1}) follows from the definition of the distance function and the use of the characteristic function of the path between vertices to expand the summation over all edges of the graph. Next, (\ref{lem1_2}) and (\ref{lem1_3}) reorder the summation and express the vertex set in terms of the partitions defined above by edge deletion. We have that $1_{\pi(i,j)}(e)=1$ if and only if the vertices $i$ and $j$ belong to distinct partitions $T_e$ and $T'_e$, from which (\ref{lem1_4}) follows. Finally, in (\ref{lem1_5}) we condense the summation notation by reordering the last sum and grouping terms. Taking the minimum over all $M \in \Gamma(P,Q)$ yields the earth mover's distance on the left hand side, and thus the desired result is obtained. 
\end{proof}

Next, we prove a lower bound on the summands involved in the above definition of the earth mover's distance.

\begin{lemma}
 For any flow $M \in \Gamma(P.Q)$ and any $e \in E$ we have that $$\sum_{i\in T_e}\sum_{j \in T'_e}l(e)(M_{i,j}+M_{j,i}) \geq l(e)\left|\sum_{i \in T_e}P(i)-Q(i)\right|.$$ Further, the differential abundance vector, indexed by the edges of $T$ and having entries ${\rm DiffAbund}_e = l(e)\sum_{i \in T_e}\sum_{j \in T_{e'}} M_{i,j}-M_{j,i}$ is unique, regardless of the minimizing flow $M$.
 \end{lemma}

\begin{proof}
We have that 
\begin{align}
l(e)\left|\sum_{i \in T_e} P_i-Q_i\right| &= \left|l(e)\sum_{i \in T_e}\left(\sum_{j \in T}M_{i,j}-\sum_{j \in T}M_{j,i}\right)\right|\label{lem2_1}\\
	&= \left|\sum_{i \in T_e}l(e)\sum_{j \in T}M_{i,j}-M_{j,i}\right|\label{lem2_2}\\
	&= \left|\sum_{i \in T_e}\left(\sum_{j \in T_e}l(e)(M_{i,j}-M_{j,i})+\sum_{j \in T'_e}l(e)(M_{i,j}-M_{j,i})\right)\right|\label{lem2_3}\\
		&= \left|\sum_{i \in T_e}\sum_{j \in T_e}l(e)(M_{i,j}-M_{j,i})+\sum_{i \in T_e}\sum_{j \in T'_e}l(e)(M_{i,j}-M_{j,i})\right|\label{lem2_4}\\
	&= \left|\sum_{i \in T_e}\sum_{j \in T'_e}l(e)(M_{i,j}-M_{j,i})\right|\label{lem2_5}\\
	&\leq \sum_{i \in T_e}\sum_{j \in T'_e}l(e)(M_{i,j}+M_{j,i})\label{lem2_6}.
\end{align}

Equations (\ref{lem2_1}) and (\ref{lem2_2}) above follow from expanding $P_i$ and $Q_i$ in terms of the row and column sums of $M$. Equations (\ref{lem2_3}) and (\ref{lem2_4}) reorganize the inner sums by way of the partitions $T_e$ and $T'_e$ and then group terms. Next we note that $\sum_{i \in T_e}\sum_{j \in T_e}l(e)(M_{i,j}-M_{j,i}) = 0$, as each term $M_{i,j}$ occurs precisely twice, once with each sign, which is reflected in (\ref{lem2_5}) above. This line also demonstrates the uniqueness of ${\rm DiffAbund}_e$, as the quantity is here shown to be equal to $\sum_{i \in T_e} P_i-Q_i$, which depends on the distributions $P$ and $Q$. Finally, we apply the triangle inequality to yield our result. 
\end{proof}

What follows is a brief technical lemma used to prove that the matrix $M$ produced by\\ EMDUnifrac is indeed a flow.

\begin{lemma}
 Let $m \in T$ be arbitrary. Then for all $n \in T$ such that $n$ is a vertex along the path from $m$ to $\rho$, when $i = n$ in the loop beginning at line 1 of Algorithm (\ref{alg:EMDUnifrac}) we have that one of the following hold:

If $m$ is a source, then at the beginning of line 4 of algorithm \ref{alg:EMDUnifrac} we have that 

\begin{align*}
P_m &= G_{n,m} + \sum_{k \in T} M_{m,k}\\
Q_m &= \sum_{k_ \in T} M_{k,m}.
\end{align*} 

Alternately, if $m$ is a sink, then at the beginning of line 4 of Algorithm (\ref{alg:EMDUnifrac}) we have that

\begin{align*}
P_m &= \sum_{k \in T} M_{m,k} \\
Q_m &= -G_{n,m} + \sum_{k \in T} M_{k,m}.
\end{align*}
\end{lemma}

\begin{proof}
This follows by induction. Suppose $m$ is a source and let $i=m$ in the loop at line 1 of Algorithm \ref{alg:EMDUnifrac}. Then $\min(P_m,Q_m) = Q_m$ and hence, by construction, $M_{m,m} = Q_m, G_{m,m} = P_m - Q_m$. Further, before beginning the loop at line 4 of Algorithm \ref{alg:EMDUnifrac}, every other entry of the $m$-th row of $M$ and $G$ are zero. This is because the elements of these rows are first potentially assigned non-zero values for $i=m$ in the midst of lines 6, 7 or 8. Thus at the beginning of line 4 of Algorithm \ref{alg:EMDUnifrac}, we have
 
\begin{align*}
P_m &= G_{m,m} + \sum_{k \in T} M_{m,k}, \\
Q_m &= \sum_{k \in T} M_{k,m}.
\end{align*}
Thus the claim holds for $i=m$.

Now suppose inductively that the above equalities holds when $i=j$ for some vertex $j \geq m$ on the path from $m$ to $\rho$ in $T$. We shall show the equalities holds for $i=a(j)$. As Algorithm \ref{alg:EMDUnifrac} proceeds in the loop at line 1 to the vertex for $i = a(j)$, we have that $G_{a(j),m} \geq 0$ and thus by line 5 of Algorithm \ref{alg:EMDUnifrac}, the $m$-th column of $M$ is left unchanged. Hence the sum $\sum_{k \in T} M_{k,m}$ remains unchanged.

Additionally, any change to $G_{a(j),m}$ during the loop at line 5 is compensated by a change to $\sum_{k \in T} M_{m,k}$, thus $$G_{a(j),m} + \sum_{k \in T} M_{m,k} = G_{j,m} + \sum_{k \in T} M_{m,k} = P_m.$$ Thus, inductively, the claims holds for all vertices along the path from $m$ to $\rho$ in $T$ and $m$ a source. Symmetric reasoning holds for the case of $m$ a sink.  
\end{proof}

We now prove our main result.

\begin{theorem}
The EMDUnifrac algorithm in Algorithm \ref{alg:EMDUnifrac} produces the earth mover's distance $EMDUnifrac(P,Q)$ and a corresponding minimizing flow $M$.
\end{theorem}

\begin{proof}
 We first show that $M$ is indeed a flow. Upon the algorithm reaching the root $\rho$, that is when $i=|T|$ in line 4 of Algorithm \ref{alg:EMDUnifrac}, we have traversed every vertex of $T$, so that 

\begin{align}
0 &= 1-1\label{prf1_1}\\
  &= \sum_{k \in T} P_k - Q_k \label{prf1_2}\\
	&= \sum_{k \in T_{source}}P_k - Q_k + \sum_{k \in T_{sink}}P_k - Q_k\label{prf1_=3}\\
	&= \sum_{k \in T_{source}}\left(G_{\rho,k} + \sum_{l \in T} M_{k,l} - \sum_{l \in T}M_{l,k}\right) + \sum_{k \in T_{sink}}\left(\sum_{l \in T} M_{k,l} - \left(-G_{\rho,k} + \sum_{l \in T}M_{l,k}\right)\right)\label{prf1_4}\\
	&= \sum_{k \in T }\sum_{l \in T}M_{l,k} - \sum_{k \in T }\sum_{l \in T}M_{k,l} + \sum_{k \in T_{source}}G_{\rho,k} + \sum_{k \in T_{source}}G_{\rho,k}\label{prf1_5}\\
	&= \sum_{k \in T} G_{\rho,k}. \label{prf1_6}
\end{align}
The above equalities are justified as follows. In (\ref{prf1_4}) we expand the terms $P_k$ and $Q_k$ in terms of the matrices $G$ and $M$, as shown in Lemma $3$, since $\rho$ is an element of the path from any vertex to $\rho$. We then group terms in (\ref{prf1_5}) and (\ref{prf1_6}) by repeatedly using that $T_{source} \cup T_{sink} = T$, before canceling the symmetric summations of the elements of $M$. 

It then follows that the sum of the positive elements of $G_{\rho,\cdot}$ is equal to the sum of the negative elements of $G_{\rho,\cdot}$, and thus, by construction of the loops at lines 4 and 5 of Algorithm \ref{alg:EMDUnifrac}, the algorithm must terminate with $G_{\rho,\cdot}$ identically zero. As we still have that for each $i \in T$, $P_i = \sum_{k \in T} M_{j,k}, Q_i = \sum_{k \in T} M_{k,j}$, up to the addition or subtraction of $G_{\rho,i}=0$, $M$ must be a flow.

Now we show that $M$ minimizes the sum defining the earth mover's distance. By Lemmas $1$ and $2$, it suffices to show that $\sum_{i \in T_e}\sum_{j \in T'_e}l(e)(M_{i,j}+M_{j,i}) = |\sum_{i \in T_e}P_i-Q_i|$ for all $e \in E$. Given the ordering of the vertices chosen for the algorithm above, let $n \in T-\{\rho\}$ be arbitrary. To begin, we make some observations regarding the structure of the matrix $G$ and its relationship to $M$ in the algorithm. Note, that by construction, at the termination of the loop at line 4 of Algorithm \ref{alg:EMDUnifrac} for $i=n$, the entries of $G_{n,\cdot}$ all have the same sign, as the the loops at lines 4 and 5 have the effect of pairwise choosing elements of opposite signs and using one to eliminate the other. This process terminates when elements of one or the other sign are exhausted. Second, note that for $k \in T'_{e_n}$ and $m > n$, either $G_{m,k} = 0$ or has the same sign as $G_{n,k}$, as any change to the entries of $G_{\cdot, k}$ is made to move the value toward zero by a quantity bounded by the magnitude of the entry. This again follows from examination of the inner most loop of the algorithm, as well as the evolution of rows of $G$. Finally, note that across all $i \in T'_{e_n}, j \in T_{e_n}$ either $M_{j,i} = 0$ or $M_{j,i} = 0$. This follows since $M_{i,j}$, respectively $M_{j,i}$, is only assigned a non-zero value in the case of $G_{m,i} > 0$, respectively $G_{m,i} < 0$. By the above observation regarding the signs of the elements of $G_{n,\cdot}$, only one of these conditions holds across $i,j$.

Now, without loss of generality, assume $$\left|\sum_{i \in T_{e_n}} P_i - Q_i\right| = \sum_{i \in T_{e_n}} P_i - Q_i$$ as the argument for the alternate case is analogous. We then have that 
\begin{align}
\left|\sum_{i \in T_{e_n}} P_i - Q_i\right|  &= \sum_{i \in T_{e_n}} P_i - Q_i\\
	&= \sum_{i \in T_{e_n}} \sum_{j \in T} M_{i,j} - M_{j,i}\\
	&= \sum_{i \in T_{e_n}} \sum_{j \in T'_{e_n}} M_{i,j} - M_{j,i}\label{prf2_3}\\
	&= \sum_{i \in T_{e_n}} \sum_{j \in T'_{e_n}} M_{i,j} + M_{j,i}.\label{prf2_4}
\end{align}
The change of sign in moving from (\ref{prf2_3}) to (\ref{prf2_4}) follows from the above observation that at least one of $M_{i,j}$ or $M_{j,i}$ must be identically zero, and that the sum must be non-negative. Hence $-M_{j,i} = 0 = M_{j,i}$. Scaling the above equality by $l({e_n})$ yields $$\left|\sum_{i \in T_{e_n}} P_i - Q_i\right| = \sum_{i \in T_{e_n}} \sum_{j \in T'_{e_n}} M_{i,j} + M_{j,i}.$$ Having achieved the lower bound established in Lemma 2, we must have that the flow $M$ is a minimizer for the sum defining ${\rm EMDUnifrac}(P,Q)$. 
\end{proof}

\begin{theorem}
\label{thm:speed}
 Let $|{\rm supp }\ P|,|{\rm supp }\ Q|$ denote the number of elements in the support of the probability distributions $P$ and $Q$, respectively. Let $s = |{\rm supp }\ P|+|{\rm supp }\ Q|$. Then the EMDUnifrac algorithm has time and space complexity $O(s)$.
 \end{theorem}

\begin{proof}
We first consider the time complexity of EMDUnifrac. Note that each iteration of the loop at line 5 of Algorithm \ref{alg:EMDUnifrac} has the effect of satisfying a source $i$ or sink $j$, that is, establishing the appropriate row sum $i$ or column sum $j$ of the matrix $M$. Further, the loop at line 5 only visits a pair of vertices $(i,j)$ in the case that both source $i$ and sink $j$ have not been satisfied, that is, that both $P(i) \neq \sum_{k\in T} M_{i,k}$ and $Q(i) \neq \sum_{k\in T} M_{k,i}$. As there are $s$ such row or column sums to satisfy, the loop at line 5 is evaluated at most $s$ times. Hence the time complexity of the algorithm is, in total, linear in $s$.

Now we examine the space requirements of EMDUnifrac. By the above, the matrix $M$ is sparse. That is, there are most $s$ evaluations of the loop at line 5 of Algorithm \ref{alg:EMDUnifrac} and thus, including the assignment of values to $M$ at line 2 of the algorithm, at most $2s$ non-zero entries in $M$. Additionally, line 3 of the algorithm assigns a non-zero entry to $G$ at most $n$ times, while line 12 has the effect of passing non-zero entries of $G$ from one row to another prior to being removed in line 13. Thus the number of non-zero entries of $G$ is bounded by $s$. Finally, the vector $w$ in Algorithm \ref{alg:EMDUnifrac} is one dimensional, having at most $s$ non-zero entries. Hence the total space requirements of the algorithm are also linear in $s$. 
\end{proof}

\section{Conclusion}


This paper implements the ideas of \cite{evans2012phylogenetic} to capitalize on the characterization of the Unifrac distance as the earth mover's distance on weighted phylogenetic trees. The EMDUnifrac algorithm developed, and proved correct, allows for extremely rapid computation of weighted and unweighted Unifrac distances between biological communities. In particular, computations times are much faster than FastUnifrac when producing identical outputs, as seen in Figure \ref{fig:timing}. These very rapid computation times and the minimal storage requirements, both linear in the number of taxa present, allow for all pairwise comparisons in large-scale studies. An example of this sort of implementation is seen in Figure \ref{fig:RealFlow}.

In addition to the Unifrac distance, EMDUnifrac is capable of producing both a minimizing flow and a differential abundance vector. The minimizing flow and differential abundance vector can be viewed as partitions of the numeric Unifrac distance, partitions which describe how operational taxonomic units present in biological communities contribute to their measured dissimilarity. The results shown in Figure \ref{fig:RealFlow} demonstrate an application in which the raw Unifrac value has less apparent discerning power than achieved by an analysis of the differential abundance vector.

Finally, EMDUnifrac algorithm is capable of computing the Unifrac distance for any weighted tree, not merely those trees weighted at their leaves. This allows for the comparison of whole genome shotgun metagenomes, an application in which weights are assigned at various levels of phylogenetic specificity. This is a capability apparently lacking in FastUnifrac, which combined with the ability to produce differential abundance vectors gives EMDUnifrac broader utility than current computational tools for measuring Unifrac distances.

The EMDUnifrac algorithm itself is an extension of the ideas presented in the \cite{mangul2015reference} which considered De Bruijn graphs. Both leverage the earth mover's distance to compute biologically relevant metrics on graphs. In EMDUnifrac, the topological benefits of a tree are exploited to speed computation in ways which are not possible under the more complicated topology of a De Bruijn graph.


\section*{Acknowledgement}

\paragraph{Funding:} None.

%
%

\bibliographystyle{abbrv}
\bibliography{bibliography}

\end{document}